\documentclass[11pt,letterpaper]{article}
\usepackage[margin=1in]{geometry}
\usepackage{amsmath,amsthm,amssymb}
\usepackage{xspace}
\usepackage{url}
\usepackage{hyperref}
\usepackage{appendix}

\usepackage{algorithm}  
\usepackage{algpseudocodex}  

\usepackage{multirow}
\usepackage{microtype}
\usepackage{booktabs}
\usepackage{graphicx}
\usepackage{subcaption}

\hypersetup{colorlinks=true,allcolors=blue}

\theoremstyle{plain}
\newtheorem{theorem}{Theorem}[section]
\newtheorem{lemma}[theorem]{Lemma}

\newtheorem{corollary}[theorem]{Corollary}

\newtheorem{fact}[theorem]{Fact}

\theoremstyle{definition}

\theoremstyle{remark}

\newcommand{\ProblemName}[1]{\textsc{#1}}
\newcommand{\kzC}{\ProblemName{$(k, z)$-Clustering}\xspace}
\newcommand{\kMedian}{\ProblemName{$k$-Median}\xspace}
\newcommand{\kMeans}{\ProblemName{$k$-Means}\xspace}

\newcommand{\calH}{\ensuremath{\mathcal{H}}\xspace}
\newcommand{\kMeanspp}{\ensuremath{\textsc{$k$-Means++}}\xspace}
\newcommand{\sklearn}{Scikit-learn\xspace}

\DeclareMathOperator{\poly}{poly}
\DeclareMathOperator{\polylog}{polylog}
\DeclareMathOperator{\cost}{cost}

\DeclareMathOperator{\dist}{dist}
\DeclareMathOperator{\Span}{span}
\DeclareMathOperator{\NN}{NN}

\begin{document}

\title{Coresets for Kernel Clustering}

\author{Shaofeng H.-C. Jiang\thanks{
    Work partially supported a startup fund from Peking University
    and the Advanced Institute of Information Technology, Peking University.
    Email: \texttt{shaofeng.jiang@pku.edu.cn}
}\\
Peking University
\and Robert Krauthgamer\thanks{Work supported partially by ONR Award N00014-18-1-2364, 
the  Israel Science Foundation grant \#1336/23, a Weizmann-UK Making 
Connections Grant, a Minerva Foundation grant, the Weizmann Data Science 
Research Center, and a research grant from the Estate of Harry Schutzman.
    Email: \texttt{robert.krauthgamer@weizmann.ac.il}
  }\\
Weizmann Institute of Science
\and Jianing Lou\thanks{Email: \texttt{loujn@pku.edu.cn}}\\
Peking University
\and Yubo Zhang\thanks{Email: \texttt{1800012916@pku.edu.cn}}\\
Peking University
}

\maketitle

\begin{abstract}
We devise coresets for kernel \kMeans with a general kernel,
and use them to obtain new, more efficient, algorithms.
Kernel \kMeans has superior clustering capability compared to classical \kMeans,
particularly when clusters are non-linearly separable,
but it also introduces significant computational challenges.
We address this computational issue by constructing a coreset, 
which is a reduced dataset that accurately preserves the clustering costs.

Our main result is a coreset for kernel \kMeans
that works for a general kernel and has size $\poly(k\epsilon^{-1})$.
Our new coreset both generalizes and greatly improves all previous results;
moreover, it can be constructed in time near-linear in $n$.
This result immediately implies new algorithms for kernel \kMeans,
such as a $(1+\epsilon)$-approximation in time near-linear in $n$,
and a streaming algorithm using space and update time $\poly(k \epsilon^{-1} \log n)$. 

We validate our coreset on various datasets with different kernels.
Our coreset performs consistently well,
achieving small errors while using very few points.
We show that our coresets can speed up kernel \kMeanspp
(the kernelized version of the widely used \kMeanspp algorithm),
and we further use this faster kernel \kMeanspp for spectral clustering.
In both applications,
we achieve significant speedup and a better  asymptotic growth
while the error is comparable to baselines that do not use coresets.
\end{abstract}

\section{Introduction}
The \kMeans problem has proved to be fundamental for unsupervised learning
in numerous application domains. 
Vanilla \kMeans fails to capture sophisticated cluster structures,
e.g., when the clusters are separable non-linearly, 
but this can be tackled by applying kernel methods 
\cite{DBLP:journals/neco/ScholkopfSM98,DBLP:journals/tnn/Girolami02}.
This has led to kernel \kMeans,
where data points are first mapped to a high-dimensional feature space
(possibly implicitly via a kernel function), 
and then clustered in this richer space using a classical \kMeans.

Formally, a \emph{kernel} for a dataset $X$
is a function $K : X \times X \to \mathbb{R}$
(intended to measure similarity between elements in $X$) 
that can be realized by inner products,
i.e., there exist a Hilbert space $\calH$ and a map $\varphi : X \to \calH$
(called feature space and feature map) such that
\begin{equation}
  \label{eqn:varphi}
  \forall x, y \in X,
  \quad
  \langle \varphi(x), \varphi(y) \rangle = K(x, y).
\end{equation}
In \emph{kernel \kMeans},
the input is a dataset $X$ with weight function $w_X : X \to \mathbb{R}_+$
and a kernel function $K : X \times X \to \mathbb{R}$ as above,
and the goal is to find a $k$-point center set $C \subseteq \calH$
that minimizes the objective
 \begin{equation}
  \label{equ:kmeans_obj}
  \cost^\varphi(X, C) = \sum_{x \in X}{ w_X(x) \cdot \min_{c \in C} \| \varphi(x) - c \|^2 }.
\end{equation}
(An equivalent formulation asks for a $k$-partitioning of $X$, keeping $C$ implicit.)

This kernel version has superior clustering capability compared to classical \kMeans~\cite{DBLP:conf/icpr/ZhangR02,DBLP:journals/pr/KimLLL05},
and has proved useful in different application domains,
such as pattern recognition, natural language processing and social networks. In fact, kernel \kMeans is useful also for solving other clustering problems, such as normalized cut and spectral clustering~\cite{DBLP:conf/kdd/DhillonGK04,DBLP:conf/ecml/DingHS05}.

\paragraph{Computational challenges.}
As observed in previous work~\cite{DBLP:journals/tnn/Girolami02,DBLP:conf/kdd/DhillonGK04},
the \emph{kernel trick} can be applied to rewrite kernel \kMeans
using access only to the kernel $K(\cdot,\cdot)$
and without computing the very high-dimensional map $\varphi$ explicitly.
However, this approach has outstanding computational challenges
(compared to classical \kMeans), essentially because of the kernel trick.
Consider the special case where $k = 1$ and the input is $n$ unweighted points
(i.e., \textsc{$1$-Mean} clustering).
It is well known that the optimal center $c^\star$ has a closed form
$c^\star := \frac{1}{n} \sum_{x \in X}{\varphi(x)}$.
But the kernel trick requires $\Omega(n^2)$ accesses to $K$ 
to evaluate $\cost^\varphi(X, c^\star)$,\footnote{In fact, evaluating $\|c^\star - \varphi(u)\|^2$
  for a single point $u \in X$ already requires $\Theta(n^2)$ accesses, 
  since $\|c^\star - \varphi(u)\|^2 =
  K(u, u) - \frac{2}{n} \sum_{x \in X}{K(x, u)} +
  \frac{1}{n^2} \sum_{x,y \in X}{K(x, y)} $.
}
while in the classical setting such evaluation
takes only $O(n)$ time.

This $\Omega(n^2)$ barrier can be bypassed at the cost of $(1+\epsilon)$-approximation.
In particular, let $S$ be a uniform sample of $\poly(\epsilon^{-1})$ points from $X$, 
and let $\hat{c} := \frac{1}{\vert S\vert}\sum_{x \in S}{\varphi(x)}$ be its \textsc{$1$-Mean};
then with high probability,
$\cost^\varphi(X, \hat{c}) \leq (1 + \epsilon)\cost^\varphi(X, c^\star)$
and evaluating $\cost^\varphi(X, \hat{c})$ takes only $\poly(\epsilon^{-1})n$ time. 
However, this uniform-sampling approach does not generally work for $k \geq 2$,
because if the optimal clustering is highly imbalanced,
a uniform sample is unlikely to include any point from a small cluster.
Alternative approaches, such as dimension reduction, were also
proposed to obtain efficient algorithms for kernel \kMeans,
but they too do not fully resolve the computational issue.
We elaborate on these approaches in Section~\ref{sec:previous}.

\paragraph{Our approach.}
To tackle this computational challenge,
we adapt the notion of coresets~\cite{DBLP:conf/stoc/Har-PeledM04}
to kernel \kMeans.
Informally, a coreset is a tiny reweighted subset of the original dataset on which 
the clustering cost is preserved within $(1\pm \epsilon)$-factor for all candidate centers $C \subseteq \calH$.
This notion has proved very useful for classical \kMeans,
e.g., to design efficient near-linear algorithms.
In our context of kernel \kMeans,
a coreset of size $s$ for an input of size $n$ has a huge advantage
that each of its $k$ optimal centers can be represented
as a linear combination of only $s$ points in the feature space.
Given these $k$ optimal centers (as linear combinations),
evaluating the distance between a point $\varphi(x)$
and a center takes merely $O(s^2)$ time, instead of $O(n^2)$, 
and consequently
the objective can be $(1+\epsilon)$-approximated in time $O(s^2 k n)$.
Moreover, it suffices to use $k$ centers (again as linear combinations)
that are $(1+\epsilon)$-approximately optimal for the coreset $S$.

In addition, coresets are very useful in dealing with massive datasets,
since an offline construction of coresets usually implies 
streaming algorithms~\cite{DBLP:conf/stoc/Har-PeledM04},
distributed algorithms~\cite{DBLP:conf/nips/BalcanEL13}
and dynamic algorithms~\cite{DBLP:conf/esa/HenzingerK20}
via the merge-and-reduce method~\cite{DBLP:conf/stoc/Har-PeledM04},
and existing (offline) algorithms can run efficiently on the coreset,
instead of the original dataset, with minor or no modifications.

Designing coresets for clustering problems has been an active area,
but unfortunately only limited results are currently known
regarding coresets for kernel clustering.
Indeed, all existing results achieve a weaker notion of coreset, apply only to some kernels, and/or have an exponential coreset size (see comparison below).
A major recent success in the area has been
the design coresets for \emph{high-dimensional} Euclidean inputs,
referring to coresets of size $\poly(k\epsilon^{-1})$,
which is \emph{independent of the Euclidean dimension}~\cite{DBLP:journals/siamcomp/FeldmanSS20,DBLP:conf/focs/SohlerW18,DBLP:conf/stoc/BecchettiBC0S19,DBLP:conf/stoc/HuangV20,DBLP:conf/soda/BravermanJKW21,DBLP:conf/stoc/Cohen-AddadSS21,DBLP:conf/icml/FengKW21}.
Intuitively, these new coresets are an excellent fit for kernels,
which are often high-dimensional, or even infinite-dimensional,
but employing these results requires a formal treatment
to streamline technical details like infinite dimension or accesses to $K$.

\subsection{Our Results}
\label{sec:our_reuslt}

We devise a coreset for kernel \kMeans
  that works for a general kernel function and has size $\poly(k\epsilon^{-1})$.
  Our new coreset significantly improves over the limited existing results,
  being vastly more general and achieving much improved bounds. 
(In fact, it also generalizes to kernel \kzC, see Section~\ref{sec:prelim} for definitions.)
Formally, an \emph{$\epsilon$-coreset} for kernel \kMeans with respect to weighted dataset $X$
and kernel function $K : X \times X \to \mathbb{R}$
is a weighted subset $S \subseteq X$,
such that for every feature space $\calH$ and feature map $\varphi$ that realize $K$, as defined in~\eqref{eqn:varphi}, 
\begin{align}
  \label{eqn:coreset}
  \forall C \subseteq \calH, \vert C\vert = k, \quad
  \cost^\varphi(S, C) \in (1\pm \epsilon) \cdot \cost^\varphi(X, C)
\end{align}

Throughout, we assume an oracle access to $K$ takes unit time,
and therefore our stated running times also bound the number of accesses to $K$.
We denote $\tilde{O}(f) = O(f\cdot\polylog f)$ to suppress logarithmic factors.

\begin{theorem}[Informal version of Theorem~\ref{thm:main}]
\label{thm:informal}
Given $n$-point weighted dataset $X$, 
oracle access to a kernel $K : X \times X \to \mathbb{R}$,
integer $k \geq 1$ and $0 < \epsilon < 1$,
one can construct in time $\tilde{O}(nk)$,
a reweighted subset $S\subseteq X$ of size $\vert S\vert = \poly(k\epsilon^{-1})$,
that with high probability is an $\epsilon$-coreset for kernel \kMeans with respect to $X$ and $K$.
\end{theorem}

\paragraph{Comparison to previous coresets.}
A \emph{weak} coreset for kernel \kMeans was designed in~\cite{DBLP:conf/compgeom/FeldmanMS07};
  in this notion, the objective is preserved only for certain candidate centers,
  whereas \eqref{eqn:coreset} guarantees this for all centers.
  Moreover, that result pertains only to finite-dimensional kernels. 
  Strong coresets of size $k^{\poly(\epsilon^{-1})}$
  were designed in~\cite{DBLP:journals/siamcomp/FeldmanSS20,schmidt_thesis},
  by computing \kMeans recursively to obtain a total of $k^{\poly(\epsilon^{-1})}$ clusters,
  and taking the geometric mean of each cluster.
  A similar size bound was obtained in~\cite{DBLP:journals/algorithms/BargerF20}
  using a slightly different approach.\footnote{Some of these results do not explicitly mention kernel \kMeans, but their method is applicable.}
  Compared with our results,
  their coreset-size bounds are exponentially larger (in $\epsilon$)
  and their construction time is worse (super-linear in $k$) .
  In addition, their guarantee is weaker
  as it introduces a fixed \emph{shift} $D > 0$ such that
  $\cost^{\varphi}(S, C) + D \in (1 \pm \epsilon) \cdot \cost^{\varphi}(X, C)$,
  whereas in our result $D = 0$.
  On top of that, all previous results work only for \kMeans,
  while our result works for the more general \kzC,
  which in particular includes the well-known variant \kMedian.

\paragraph{Technical overview.}
  Technically, our result builds upon the recent
  coresets for finite-dimensional Euclidean spaces,
  which have dimension-independent size bound $\poly(k\epsilon^{-1})$.
  In particular, we rely on coresets that can be constructed in \emph{near-linear time} even in the kernel setting (via the kernel trick)~\cite{DBLP:conf/soda/BravermanJKW21}.
  In contrast, another recent algorithm~\cite{DBLP:conf/focs/SohlerW18} is slow and uses the explicit representation in the feature space.
  To adapt the coreset of~\cite{DBLP:conf/soda/BravermanJKW21} to the kernel setting,
  the main technical step is an embedding from (infinite-dimensional) Hilbert space
  to finite-dimensional Euclidean space, which eliminates any potential dependence on the ambient space (e.g., which points can serve as centers).
  This step requires a formal treatment, e.g., careful definitions,
  although its proof is concise. 
  Crucially, we apply this embedding step (which could be computationally heavy) only in the analysis, avoiding any blowup in the running time.
In hindsight, all these technical details fit together very smoothly, but this might be less obvious apriori; 
  moreover, the simplicity of our construction is a great advantage for implementation.

\paragraph{FPT-PTAS for kernel \kMeans.}
We can employ our coreset to devise
a $(1+\epsilon)$-approximation algorithm for kernel \kMeans,
that runs in time that is near-linear in $n$ and parameterized by $k$.
This is stated in Corollary~\ref{cor:fpt_ptas}, 
whose corresponding algorithm, presented in Algorithm~\ref{alg:PTAS}, 
computes a coreset $S$ and solves \kMeans on $S$ optimally
by straightforward enumeration over all $k$-partitions of $S$. 
Previously, $(1+\epsilon)$-approximation for kernel \kMeans has been
achieved by \cite{DBLP:conf/focs/KumarSS04},\footnote{\cite{DBLP:conf/focs/KumarSS04} developed an FPT-PTAS only for vanilla (i.e., not kernel) \kMeans, but it be adapted to kernel \kMeans by representing a center as a linear combination of points in $\varphi(X)$. }whose algorithm runs in time $O(n \cdot 2^{\poly(k\epsilon^{-1})})$.
Our algorithm is more efficient,
as it isolates the data size $n$ from the factor $2^{\poly(k\epsilon^{-1})}$.

\begin{corollary}[FPT-PTAS]
\label{cor:fpt_ptas}
Given $n$-point weighted dataset $X$,
oracle access to a kernel $K : X \times X \to \mathbb{R}$,
integer $k \geq 1$ and $0 < \epsilon < 1$,
one can compute in time $O(nk + k^{\poly(k\epsilon^{-1})})$, 
a center set $C$ of $k$ points,
each represented as a linear combination of 
at most $\poly(k\epsilon^{-1})$ points from $\varphi(X)$,
such that with high probability $C$ is a $(1+\epsilon)$-approximation 
for kernel \kMeans on $X$ and $K$. 
In particular, given such $C$, one can find
for each $x\in X$ its closest center in $C$ in time $\poly(k\epsilon^{-1})$.
\end{corollary}

\begin{algorithm}[t]
  \caption{FPT-PTAS for kernel \kMeans on dataset $X$ with kernel $K$}
  \label{alg:PTAS} 
  \begin{algorithmic}[1]
    \State construct $\epsilon$-coreset $S$ of size $\poly(k\epsilon^{-1})$ for kernel \kzC on $X$ with kernel $K$
    
    \Comment{use Theorem~\ref{thm:informal} }
    \State enumerate over all $k^{\poly(k\epsilon^{-1})}$ $k$-partitions of $S$
    to find $k$-partition $\mathcal{P}=\{P_1,\dots,P_k\}$ 
    with smallest \kMeans objective 
      $$
      \sum_{i=1}^k\frac{1}{\vert P_i\vert}\sum_{x,y\in P_i}\|x-y\|^2
      $$ \Comment{compute distances using the kernel trick}
      \\
      \Return optimal center set for $\mathcal{P}$,
      i.e., $C=\{c_1,\ldots,c_k\}$,
      where each $c_i$ is represented as $\frac{1}{\vert P_i\vert}\sum_{x\in P_i}\varphi(x)$
  \end{algorithmic}
\end{algorithm}

\paragraph{Streaming algorithms.}
In fact, for the purpose of finding near-optimal solutions,
it already suffices to preserve the cost for centers coming from $\Span(\varphi(X))$ (see Fact~\ref{fact:well_define}) which is an $n$-dimensional subspace.
However, our definition of coreset in~\eqref{eqn:coreset} is much stronger, in that
the objective is preserved even for centers coming from a
possibly infinite-dimensional feature space.
This stronger guarantee ensures that the coreset is composable,
and thus the standard merge-and-reduce method can be applied.
In particular, our coreset implies the \emph{first} streaming algorithm
for kernel \kMeans.

\begin{corollary}[Streaming kernel \kMeans]
\label{cor:streaming}
There is a streaming algorithm that given a dataset $X$
presented as a stream of $n$ points,
and oracle access to a kernel $K : X \times X \to \mathbb{R}$,
constructs a reweighted subset $S \subseteq X$ of $\poly(k\epsilon^{-1})$ points
using $\poly(k\epsilon^{-1} \log n)$ words of space and update time,
such that with high probability $S$ is an $\epsilon$-coreset for \kMeans with respect to $X$ and $K$.
\end{corollary}

\paragraph{Experiments and other applications.}
We validate the efficiency and accuracy of our coresets on various datasets 
with polynomial and Gaussian kernels, which are frequently-used kernels. 
For every dataset, kernel, and coreset-size that we test,
our coreset performs consistently better than uniform sampling which serves as a baseline.
In fact, our coreset achieves less than $10\%$ error using only about $1000$ points for every dataset.

We also showcase significant speedup to several applications
that can be obtained using our coresets.
Specifically, we adapt the widely used \kMeanspp~\cite{DBLP:conf/soda/ArthurV07} to the kernel setting,
and we compare the running time and accuracy of this kernelized \kMeanspp
with and without coresets.
In these experiments, we observe a significant speedup and more importantly,
a better asymptotic growth of running time
of \kMeanspp when using coresets, while achieving a very similar error.
Furthermore, this new efficient version of kernelized \kMeanspp
(based on coresets) is applied to solve spectral clustering,
using the connection discoverd by~\cite{DBLP:conf/kdd/DhillonGK04}.
Compared to the implementation provided by \sklearn~\cite{sklearn},
our algorithm often achieves a better result and uses significantly less time.
Hence, our coreset-based approach can potentially become
the leading method for solving spectral clustering in practice.

\subsection{Comparison to Other Approaches}
\label{sec:previous}

The computational issue of kernel \kMeans
is an important research topic and has attracted significant attention.
In the following, we compare our result with previous work
that is representative of different approaches for the problem.

Uniform sampling of data points is a commonly used technique,
which fits well in kernel clustering, 
because samples can be drawn without any access to the kernel.
While the coresets that we use rely on sampling, 
we employ importance sampling, which is non-uniform by definition.
\cite{DBLP:conf/kdd/ChittaJHJ11} employs uniform sampling for kernel \kMeans, 
but instead of solving kernel \kMeans on a sample of data points directly
(as we do),
their method works in iterations, similarly to Lloyd's algorithm,
that find a center set that is a linear combination of the sample. 
However, it has no worst-case guarantees on the error or on the running time,
which could be much larger than $\tilde{O}(nk)$.
Lastly, this method does not generalize easily to other sublinear settings,
such as streaming (as our Corollary~\ref{cor:streaming}).
\cite{DBLP:conf/icpr/RenD20} analyze uniform sampling for \kMeans
in a Euclidean space, which could be the kernel's feature space. 
In their analysis, the number of samples (and thus running time)
crucially depends on the diameter of the dataset and on the optimal objective value,
and is thus not bounded in the worse-case.
This analysis builds on several earlier papers, 
for example~\cite{DBLP:journals/rsa/CzumajS07} achieve bounds of similar flavor
for \kMeans in general metric spaces.

Another common approach to speed up kernel \kMeans
is to approximate the kernel $K$ using dimension-reduction techniques.
In a seminal paper, \cite{DBLP:conf/nips/RahimiR07} proposed a method
that efficiently computes
a low (namely, $O(\log{n})$) dimensional feature map $\tilde\varphi$
that approximates $\varphi$ (without computing $\varphi$ explicitly),
and this $\tilde\varphi$ can be used in downstream applications.
Their method is based on designing random Fourier features,
and works for a family of kernels that includes the Gaussian one,
but not general kernels.
This method was subsequently tailored to kernel \kMeans
by~\cite{DBLP:conf/icdm/ChittaJJ12},
and another followup work by~\cite{DBLP:conf/alt/ChenP17}
established worst-case bounds for kernel \kMeans with Gaussian kernels.
Despite these promising advances, we are not aware of any work based on 
dimension reduction that can handle a general kernel function (as in our approach).
In a sense, these dimension-reduction techniques are ``orthogonal''
to our data-reduction approach,
and the two techniques can possibly be combined to yield even better results.

An alternative dimension-reduction approach
is low-rank approximation of the kernel matrix $K$.
Recent works by~\cite{DBLP:conf/nips/MuscoM17} and by~\cite{DBLP:journals/jmlr/WangGM19}
present algorithms based on Nystr\"om approximation
to compute a low-rank (namely, $O(k / \epsilon)$) approximation 
$\widetilde{K}:=UU^T$
to the kernel matrix $K$ in time near-linear in $n$,
where $U\in \mathbb{R}^{n\times O(k / \epsilon)}$ 
defines an embedding of data points into $\mathbb R^{O(k / \epsilon)}$.
One might then try to construct a coreset on the subspace of $U$ via 
applying \kMeans coreset construction after dimension-reduction,
however, it does not satisfy our definition, in which we require the objective 
to be preserved for all centers from any ambient space that realizes $K$,
not only a specific space induced by $U$. Such a seemingly slight difference 
plays a great role when designing algorithms for sublinear models, such as streaming,
since it ensures the coreset is mergable 
(i.e. $\mathrm{coreset}(A)\cup \mathrm{coreset}(B)$ is a coreset of $A\cup B$), 
so that the standard merge-and-reduce technique can be applied.

 \section{Preliminaries}
\label{sec:prelim}

\paragraph{Notations.}
A weighted set $U$ is a finite set $U$ associated with a weight function $w_U : U \to \mathbb{R}_+$.
For such $U$, let $\|U\|_0$ be the number of distinct elements in it.
For a weight function as above and a subset $S \subseteq U$,
define $w_U(S) := \sum_{u \in S}{w_U(u)}$.
For any other map $f : U \to V$ (not a weight function),
we follow the standard definition $f(S) := \{ f(x) : x \in S\}$.
For an integer $t \geq 1$, let $[t] := \{ 1, \ldots, t \}$.
For $x > 0$ and integer $i \geq 1$,
let $\log^{(i)}{x}$ be the $i$-th iterated log of $x$,
i.e., $\log^{(1)}{x} = \log{x}$
and $\log^{(i)}{x} = \log (\log^{(i-1)}{x})$ for $i \geq 2$.

\paragraph{Kernel functions.}
Let $X$ be a set of $n$ points.
A function $K : X \times X \rightarrow \mathbb{R}$ is a kernel function
if the $n \times n$ matrix $M$ such that $M_{ij} = K(x_i, x_j)$ (where $x_i, x_j \in X$) is positive semi-definite.
Since $M$ is positive semi-definite,
there exists a map $\varphi$ from $X$ to some Hilbert space $\mathcal{H}$,
such that all $x, y \in X$ satisfy $K(x, y) = \langle \varphi(x), \varphi(y) \rangle$.
This above (existence of a map $\varphi$ of into $\calH$) can be extended
to infinite $X$, e.g., $X=\mathbb{R}^d$, by Mercer's Theorem. 
The distance between $x',y' \in \calH$ is defined as
$\dist(x', y') := \|x' - y'\| = \sqrt{\langle x' - y', x' - y' \rangle}$. 
Hence, the distance $\dist(\varphi(x), \varphi(y))$ for $x,y \in X$
can be represented using $K$ as
\begin{align*}
    \dist(\varphi(x), \varphi(y)) = \|\varphi(x)- \varphi(y)\| = \sqrt{K(x, x) + K(y, y) - 2K(x, y)}.
\end{align*}
We refer to a survey by~\cite{DBLP:journals/corr/abs-2106-08443} for a more comprehensive introduction to kernel functions.

\paragraph{Data model.}
For a dataset $X$ associated with a kernel (as in Theorem~\ref{thm:main}), 
it is convenient to assume that $X = \{1, \ldots, n\}$
and one has oracle access to the kernel function. 
Our actual requirement is that 
every point in $X$ can be stored in $O(1)$ space and manipulated in $O(1)$ time,
and the kernel can be evaluated on any pair of points from $X$ in $O(1)$ time. 
A similar model is assumed when $X$ is associated with a distance function
(as in Theorem~\ref{thm:euc_coreset}).

\paragraph{Kernel \kzC.}
In the the kernel \kzC problem, 
the input is a weighted dataset $X$ of $n$ objects,
a kernel function $K : X \times X \to \mathbb{R}$,
an integer $k \geq 1$, and $z > 0$.
The goal is to find a $k$-point center set $C \subseteq \calH$
that minimizes the objective
\begin{align}
  \label{equ:kz_obj}
    \cost_z^\varphi(X, C) := \sum_{x \in X}{ w_X(x) (\dist(\varphi(x), C))^z },
\end{align}
where $\calH$ is an induced Hilbert space of $K$ and $\varphi : X \to \calH$ is its feature map,
and 
$\dist(\varphi(x), C) := \min_{c \in C}{\dist(\varphi(x), c)} = \min_{c \in C}{\|\varphi(x) - c\|}$.
The case $z=2$ clearly coincides with kernel \kMeans whose objective is~\eqref{equ:kmeans_obj}.
The (non-kernel) \kzC problem may be viewed as kernel \kzC with kernel $K(x, y) = \langle x, y\rangle$ and identity feature map $\varphi(x) = x$.

While the feature map $\varphi$ might not be unique,
we show below that this kernel \kzC is well defined,
in the sense that the optimal value is independent of $\varphi$.
The following two facts are standard and easy to prove.

\begin{fact}
\label{fact:well_define}
For every map $\varphi$ into $\calH$, there is an optimal solution $C^\star$
in which every center point $c\in C^\star$ lies inside $\Span(\varphi(X))$,
and is thus a linear combination of $\varphi(X)$. 
\end{fact}

\begin{corollary}
\label{cor:well_define}
The optimal value of~\eqref{equ:kz_obj} can be represented
as a function of kernel values $K(x,y)$, 
thus invariant of $\varphi$.
\end{corollary}

\paragraph{$\epsilon$-Coresets for kernel \kzC.}
For $0 < \epsilon < 1$, an $\epsilon$-coreset for kernel \kzC on a weighted dataset $X$ and a kernel function $K$
is a reweighted subset $S \subseteq X$,
such that for every Hilbert space $\calH$ and map $\varphi : X \to \calH$ 
satisfying~\eqref{eqn:varphi},
we have
\begin{align*}
    \forall C \subseteq \calH, \vert C\vert = k,
    \quad \cost_z^\varphi(S, C) \in (1 \pm \epsilon) \cdot \cost_z^\varphi(X, C).
\end{align*}
The case $z=2$ clearly coincides with~\eqref{eqn:coreset}.

 \section{Coresets for Kernel \kzC}

\begin{theorem}
    \label{thm:main}
    Given $n$-point weighted dataset $X$, 
    oracle access to a kernel $K : X \times X \to \mathbb{R}$,
    $z \geq 1, 0 < \epsilon < 1$, and integer $k \geq 1$,
    one can construct in time $\tilde{O}(nk)$,
    a reweighted subset $S\subseteq X$ of size $\|S\|_0 = 2^{O(z)} \cdot \poly(k\epsilon^{-1})$,
    that with high constant probability is an $\epsilon$-coreset for kernel \kzC with respect to $X$ and $K$.
\end{theorem}

At a high level, our proof employs 
recent constructions of coresets for \kzC in Euclidean spaces,
in which the coreset size is independent of the Euclidean dimension~\cite{DBLP:conf/focs/SohlerW18,DBLP:journals/siamcomp/FeldmanSS20,DBLP:conf/stoc/HuangV20,DBLP:conf/soda/BravermanJKW21}.
However, these coresets are designed for finite-dimensional Euclidean spaces,
and are thus not directly applicable to our feature space $\calH$,
which might have infinite dimension.

To employ these coreset constructions, we show 
that the data points in the feature space $\calH$
embed into an $(n+1)$-dimensional (Euclidean) space,
without any distortion to distances between data points and centers (Lemma~\ref{lemma:nplus1}).
This observation is similar to one previously made by~\cite{DBLP:conf/focs/SohlerW18} for a different purpose.
Due to this embedding, it suffices to construct coresets
in a limited setting 
where centers come only from $\mathbb{R}^{n+1}$ (Corollary~\ref{cor:nplus1}).

\begin{lemma}
\label{lemma:nplus1}
Let $\calH$ be a Hilbert space
and let $X \subseteq \calH$ be a subset of $n$ points.
Then there exists a map $f : \calH \to \mathbb{R}^{n+1}$
such that 
$\forall x \in X, c \in \calH,
  \quad
  \|x - c\| = \|f(x) - f(c)\|$.
\end{lemma}
\begin{proof}
  Let $\mathcal{S} = \Span(X)$.
  Then every point $c \in \calH$ can be written (uniquely) as 
  $c = c^{\|} + c^{\perp}$,
  where $c^{\|} \in \mathcal{S}$ and $c^{\perp}$ is orthogonal to $\mathcal{S}$.
  Thus, $\|c\|^2 = \|c^{\|}\|^2 + \|c^{\perp}\|^2$.
  Note that for all $x\in X$, we have $x^{\perp} = 0$.
Now, for every $c \in \calH$,
  let $f(c) := ( c^{\|} \ ; \|c^{\perp}\| )$,
  where we interpret $x^{\|}$ as an $n$-dimensional vector.
Then for all $x\in X$ and $c \in \calH$,
    \begin{align*}
        \|x - c\|^2
        &= \|x^{\|} - c^{\|}\|^2 + 
            \|x^{\perp} - c^{\perp}\|^2 \\
        &= \|x^{\|} - c^{\|}\|^2 + 
            \|c^{\perp}\|^2 \\
        &= \|f(x) - f(c)\|^2.
    \end{align*}
    The claim follows.
\end{proof}

\begin{corollary}
\label{cor:nplus1}
Consider $n$-point weighted dataset $X$,
kernel function $K : X \times X \to \mathbb{R}$,
$z \geq 1$, integer $k \geq 1$, and $0 < \epsilon < 1$.
Suppose that a reweighted subset $S \subseteq X$ satisfies that
for every $\varphi : X \to \mathbb{R}^{n+1}$
such that for all $x, y \in X$, $\langle \varphi(x), \varphi(y) \rangle = K(x, y)$, for all $C \subseteq \mathbb{R}^{n + 1}$ with $\vert C\vert = k$
the following holds
\begin{align*}
\cost_z^\varphi(S, C) \in (1 \pm \epsilon) \cdot \cost_z^\varphi(X, C) .
\end{align*}
Then $S$ is an $\epsilon$-coreset for kernel \kzC with respect to $X$ and kernel $K$.
\end{corollary}

\begin{proof}
To verify that $S$ is a coreset with respect to $X$ and $K$,
consider some feature space $\calH'$ and feature map $\varphi'$ be induced by $K$.
Apply Lemma~\ref{lemma:nplus1} to obtain $f : \calH' \to \mathbb{R}^{n+1}$,
then for all $C \subseteq \calH'$, $\vert C\vert = k$, we have
$\forall Q\subseteq X, 
\cost_z^{\varphi'}(Q, C) = \cost_z^{f \circ \varphi'}(Q, f(C)),$
and using the promise about $S$ with $\varphi = f \circ \varphi'$,
\begin{align*}
  \cost_z^{\varphi'}(S, C)
  &= \cost_z^{\varphi}(S, f(C)) \\
  &\in (1 \pm \epsilon) \cdot \cost_{z}^{\varphi}(X, f(C))  \\
  &\in (1 \pm \epsilon) \cdot \cost_{z}^{\varphi'}(X, C).
\end{align*}
Thus, $S$ is indeed a coreset with respect to $X$ and $K$.
\end{proof}

Another issue is that some of the existing algorithms, such as~\cite{DBLP:conf/focs/SohlerW18,DBLP:journals/siamcomp/FeldmanSS20},
require the explicit representations of points in $\varphi(X)$,
which is very expensive to compute.
Fortunately, the importane-sampling-based algorithms
of~\cite{DBLP:conf/stoc/HuangV20} and~\cite{DBLP:conf/soda/BravermanJKW21}
are oblivious to the representation of $\varphi$,
and only rely on a distance oracle that evaluates $\forall x, y \in X$, $\|\varphi(x) - \varphi(y)\| = \sqrt{K(x, x) + K(y, y) - 2K(x, y)}$.
Now, by Corollary~\ref{cor:nplus1}, executing these algorithms
without any modifications (except for plugging in the distance oracle defined by kernel $K$) yields the desired coreset for kernel \kzC.
We choose to use the coreset construction of~\cite{DBLP:conf/soda/BravermanJKW21}, 
which is arguably simpler. 
We now recall its statement for completeness. 

\begin{theorem}[\cite{DBLP:conf/soda/BravermanJKW21}]
    \label{thm:euc_coreset}
    Given $n$-point weighted dataset $X\subset\mathbb{R}^m$ for some interger $m$, 
    together with $z \geq 1$, integer $k \geq 1$, and $0 < \epsilon < 1$, 
    one can construct in time $\tilde{O}(mnk)$
    a reweighted subset $S\subseteq X$ of size $\|S\|_0 = \tilde{O}(\epsilon^{-4}2^{2z}k^2)$,
    that with high constant probability is an $\epsilon$-coreset for \kzC with respect to $X$.
If the explicit representation of $X\subset \mathbb{R}^m$
      is replaced by oracle access to the distance function $\dist : X \times X \to \mathbb{R}_+$,
      then the time bound improves to $\tilde{O}(nk)$.
\end{theorem}

\begin{proof}[Proof of Theorem~\ref{thm:main}]
Suppose $f$ is the asserted map from Lemma~\ref{lemma:nplus1}.
    Hence, $f(\varphi(X)) = \{ f(\varphi(x)) : x \in X \}$ is a subset of $\mathbb{R}^{n + 1}$.
    Apply Theorem~\ref{thm:euc_coreset} on top of this $f(\varphi(X))$,
    then by the guarantee of Theorem~\ref{thm:euc_coreset} and Corollary~\ref{cor:nplus1},
    the resulting weighted set is an $\epsilon$-coreset for $X$.

To show that the running time is $\tilde{O}(nk)$,
    it suffices to verify that the distance between a pair of points in $f(\varphi(X))$ can be evaluated in $O(1)$ time (as required by Theorem~\ref{thm:euc_coreset}).
    Indeed, by Lemma~\ref{lemma:nplus1},
    for all $x, y \in X$,
    one can evaluate $\dist(f(\varphi(x)), f(\varphi(y))) = \dist(\varphi(x), \varphi(y))$
    in $O(1)$ time using $3$ queries to the kernel function.
\end{proof}

\subsection{Description of Coreset Construction Algorithms}

For completeness, we present our full algorithm in Algorithm~\ref{alg:main}
(which depends on the subroutines defined in Algorithm~\ref{alg:imp_sampling} and~\ref{alg:dz_sampling}).
As can be easily seen in the proof of Theorem~\ref{thm:main},
our algorithm is a black-box application of Theorem~\ref{thm:euc_coreset},
i.e., it is identical to the algorithm of~\cite{DBLP:conf/soda/BravermanJKW21},
except that we use the kernel distance (instead of Euclidean distance),
which can be computed efficiently via the kernel trick.

The following notation is needed.
For a subset $C \subseteq \calH$ and data point $x \in X$, 
define $\NN_C(x) := \arg\min\{ \dist(\varphi(x), y) : y \in C \}$
as the nearest neighbor of $x$ in $C$ with respect to the distances in the feature space (breaking ties arbitrarily).
Thus $\NN_C(\cdot)$ defines a $\vert C\vert$-partition of $X$,
and the cluster that $x$ belongs to (with respect to $C$) is denoted 
$C(x) := \{ x' \in X : \NN_C(x') = \NN_C(x) \}$.

Algorithm~\ref{alg:main} is the main procedure for constructing the coreset,
and it iteratively executes another importance-sampling-based
coreset construction (Algorithm~\ref{alg:imp_sampling}).
Informally, each invocation of \textsc{Importance-Sampling} constructs 
a coreset $X_i$ from the current coreset $X_{i-1}$,
to reduce the number of distinct elements in $X_{i-1}$ to roughly $\log{\|X_{i-1}\|_0}$.
The procedure ends when such size reduction cannot be done any more,
at which point the size of the coreset reaches the bound in Theorem~\ref{thm:euc_coreset}, which is independent of $n$.

\begin{algorithm}[t]
    \caption{Constructing $\epsilon$-coreset for kernel \kzC on dataset $X$ with kernel $K$}
    \label{alg:main}
    \begin{algorithmic}[1]
\State let $X_0 \gets X$, $i \gets 0$
            \Repeat
                \State let $i \gets i + 1$ and $\epsilon_i \gets \epsilon / (\log^{(i)}{\|X\|_0})^{1/4}$
                \State $X_i \gets \textsc{Importance-Sampling}(X_{i-1}, \epsilon_i)$
            \Until{$\|X_i\|_0$ does not decrease compared to $\|X_{i-1}\|_0}$
            \\
\Return $X_i$
\end{algorithmic}
\end{algorithm}

In fact, subroutine \textsc{Importance-Sampling} already constructs a coreset,
albeit of a worse size that depends on $\log{\|X\|_0}$.
This subroutine is based on the well-known importance sampling approach
that was proposed and improved in a series of works, e.g., ~\cite{DBLP:conf/soda/LangbergS10,DBLP:conf/stoc/FeldmanL11,DBLP:journals/siamcomp/FeldmanSS20}.
Its first step is to compute an importance score $\sigma_x$
for every $x\in X$ (lines 1--2), and then draw independent samples from $X$
with probability proportional to $\sigma_x$ (lines 3--4).
The final coreset is formed by reweighting the sampled points (lines 5--6).
Roughly speaking, the importance score $\sigma_x$ measures
the relative contribution of $x$ to the objective function in the worst-case,
which here means the maximum over all choices of the center set.
It can be computed from an $O(\log k)$-approximate solution for kernel \kzC on $X$,
which is obtained by the $D^z$-sampling subroutine (Algorithm~\ref{alg:dz_sampling}),
a natural generalization of the $D^2$-sampling introduced by~\cite{DBLP:conf/soda/ArthurV07}.

We stress that our algorithm description uses the feature vectors $\varphi(x)$
for clarity of exposition.
These vectors do not have to be provided explicitly,
because only distances between them are required, 
and thus all steps can be easily implemented using the kernel trick,
and the total time (and number of accesses to the kernel function $K$) 
is only $\tilde{O}(nk)$.

\begin{algorithm}[t]
    \caption{\textsc{Importance-Sampling}($X$, $\epsilon$)}
    \label{alg:imp_sampling}
    \begin{algorithmic}[1]
\State let $C^\star \gets \textsc{$D^z$-Sampling}(X)$
            \State $\forall x\in X$, let $\sigma_x \gets 
            \frac{w_X(x) \cdot (\dist(x, C^\star))^z}{\cost_z^\varphi(X, C^\star)}
            + \frac{w_X(x)}{w_X(C^\star(x))}
            $ 
            \State $\forall x\in X$, let $p_x \gets \frac{\sigma_x}{\sum_{y \in X}{\sigma_y}}$
            \State draw $N = O(\epsilon^{-4} 2^{2z} z k^2 \log^2 k \log{\|X\|_0})$
            i.i.d. samples from $X$, using probabilities $(p_x)_{x \in X}$ 
            \State let $D$ be the sampled set; $\forall x \in D$ let $w_D(x) \gets \frac{w_X(x)}{p_x N}$
            \\
\Return weighted set $D$
\end{algorithmic}
\end{algorithm}

\begin{algorithm}[t]
    \caption{\textsc{$D^z$-Sampling}($X$)}
    \label{alg:dz_sampling}
    \begin{algorithmic}[1]
\State draw $x \in X$ uniformly at random, and initialize $C \gets \{\varphi(x)\}$ 
            
\Comment{
            we use here the feature vectors for clarity, and the implementation should represent $\varphi(x)$ by storing $x$ 
            }
\For{$i = 1, \ldots, k-1$}
                \State draw one sample $x \in X$, using probabilities
                $w_X(x) \cdot \frac{(\dist(\varphi(x), C))^z}{\cost_z^\varphi(X, C)}$
                \State let $C \gets C \cup \{ \varphi(x) \}$ \EndFor
            \\
\Return $C$
\end{algorithmic}
\end{algorithm}

 \section{Experiments}
\label{sec:exp}
We validate the empirical performance of our coreset for kernel \kMeans
on various datasets,
and show that our coresets significantly speed up a kernelized 
version of the widely used \kMeanspp~\cite{DBLP:conf/soda/ArthurV07}.
In addition, we apply this new coreset-based kernelized \kMeanspp to spectral clustering (via a reduction devised by~\cite{DBLP:conf/kdd/DhillonGK04}),
showing that it outperforms the well-known \sklearn solver in both speed and accuracy.

\paragraph{Experimental setup.}
Our experiments are conducted on standard clustering datasets that consist of vectors in $\mathbb{R}^d$,
and we use the RBF kernel (radial basis function kernel, also known as Gaussian kernel) and polynomial kernels as kernel functions.
An RBF kernel $K_G$ is of the form $K_G(x, y) := \exp\left( - \frac{\| x - y \|}{2\sigma^2} \right)$, where $\sigma>0$ is a parameter,
and a polynomial kernel $K_P$ is of the form $K_P(x, y) := (\langle x, y \rangle + c)^d$ where $c$ and $d$ are parameters.
Table~\ref{tab:data} summarizes the specifications of datasets and our choice of the parameters for the kernel function.
We note that the parameters are dataset-dependent, and that for Twitter and Census1990
dataset we subsample to $10^5$ points since otherwise it takes too long to run for some of our inefficient baselines.
Unless otherwise specified, we use a typical value $k = 5$ for the number of clusters.
All experiments are conducted on a PC with Intel Core i7 CPU and 16 GB memory,
and algorithms are implemented using C++.

\begin{table}[t]
    \caption{Specifications of datasets}
    \label{tab:data}
    \centering
    \small
    \begin{tabular}{llll}
        \toprule
        dataset & size & RBF kernel param. & poly. kernel param. \\
        \midrule
        Twitter~\cite{twitter_data} & 21040936 & $\sigma = 50$ & $c = 0, d = 4$ \\
Census1990~\cite{census1990_data} & 2458284 & $\sigma = 100$ & $c = 0, d = 4$ \\
Adult ~\cite{adult_data} & 48842 & $\sigma = 200000$ & $c = 0, d = 2$ \\
Bank ~\cite{bank_data} & 41188 & $\sigma = 500$ & $c = 0, d = 4$ \\
        \bottomrule
    \end{tabular}
\end{table}

\subsection{Size and Empirical Error Tradeoff}
Our first experiment evaluates the empirical error versus coreset size.
In our coreset implementation, we simplify the construction in Algorithm~\ref{alg:main}
by running the importance sampling step only once instead of running it iteratively,
and it turns out this simplification still achieves excellent performance.
As in many previous implementations, instead of setting $\epsilon$ and solving
for the number of samples $N$ in the \textsc{Importance-Sampling} procedure (Algorithm~\ref{alg:imp_sampling}), 
we simply set $N$ as a parameter to directly control the coreset size. 
We construct the coreset with this $N$
and evaluate its error by drawing 500 random center sets $\mathcal{C}$
(each consisting of $k$ points) from the dataset,
and evaluate the maximum empirical error, defined as 
\begin{equation}
    \hat{\epsilon} := \max_{C \in \mathcal{C}}\frac{\vert \cost^\varphi(S, C) - \cost^\varphi(X, C)\vert}{\cost^\varphi(X, C)}.
\end{equation}
This error is measured similarly as in the definition of coreset,
except that it is performed on a sample of center sets.
To make the measurement stable,
we independently evaluate the empirical error 100 times and report the average.

The tradeoff between the coreset size and empirical error is shown in Figure~\ref{fig:size_vs_error}, where we also compare with a baseline that constructs coresets by uniform sampling.
These experiments show that our coreset performs consistently well on all datasets and kernels.
Furthermore, our coreset admits a similar error curve regardless of dataset and
kernel function
-- for example, one can get $\leq 10\%$ error using a coreset of $1000$ points, 
--
which is perfectly justified by our theory that the size of $\epsilon$-coreset only depends on $\epsilon$ and $k$.
Comparing with the uniform-sampling baseline,
our coreset generally has superior performance,
especially when the coreset size is small.
We also observe that the uniform sampling suffers a larger variance compared with our coreset.

\begin{figure}[t]
    \centering
    \begin{subfigure}[b]{0.39\textwidth}
        \centering
        \includegraphics[width=\textwidth]{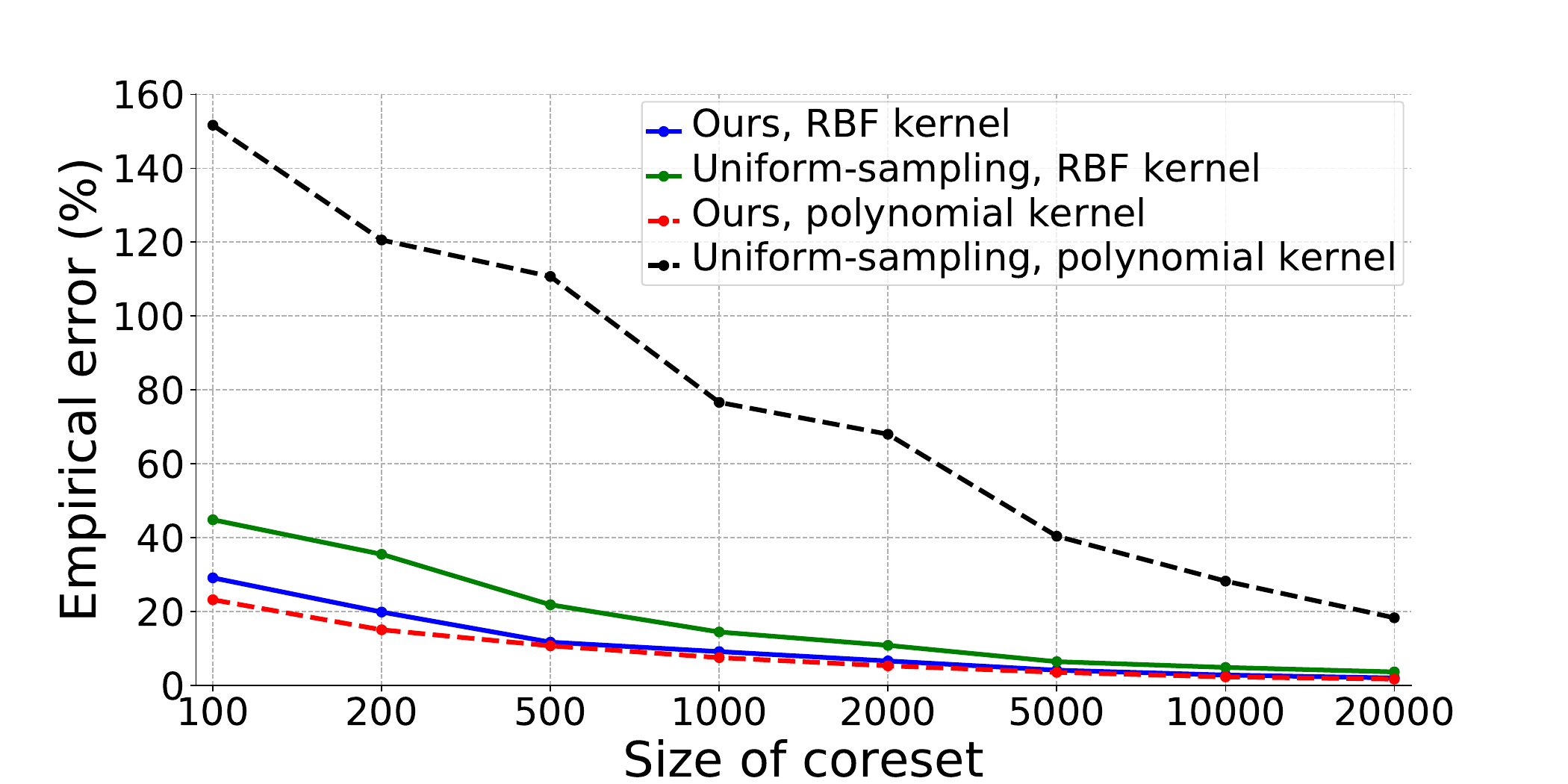}
        \caption*{Adult dataset}
    \end{subfigure} \qquad
    \begin{subfigure}[b]{0.39\textwidth}
        \centering
        \includegraphics[width=\textwidth]{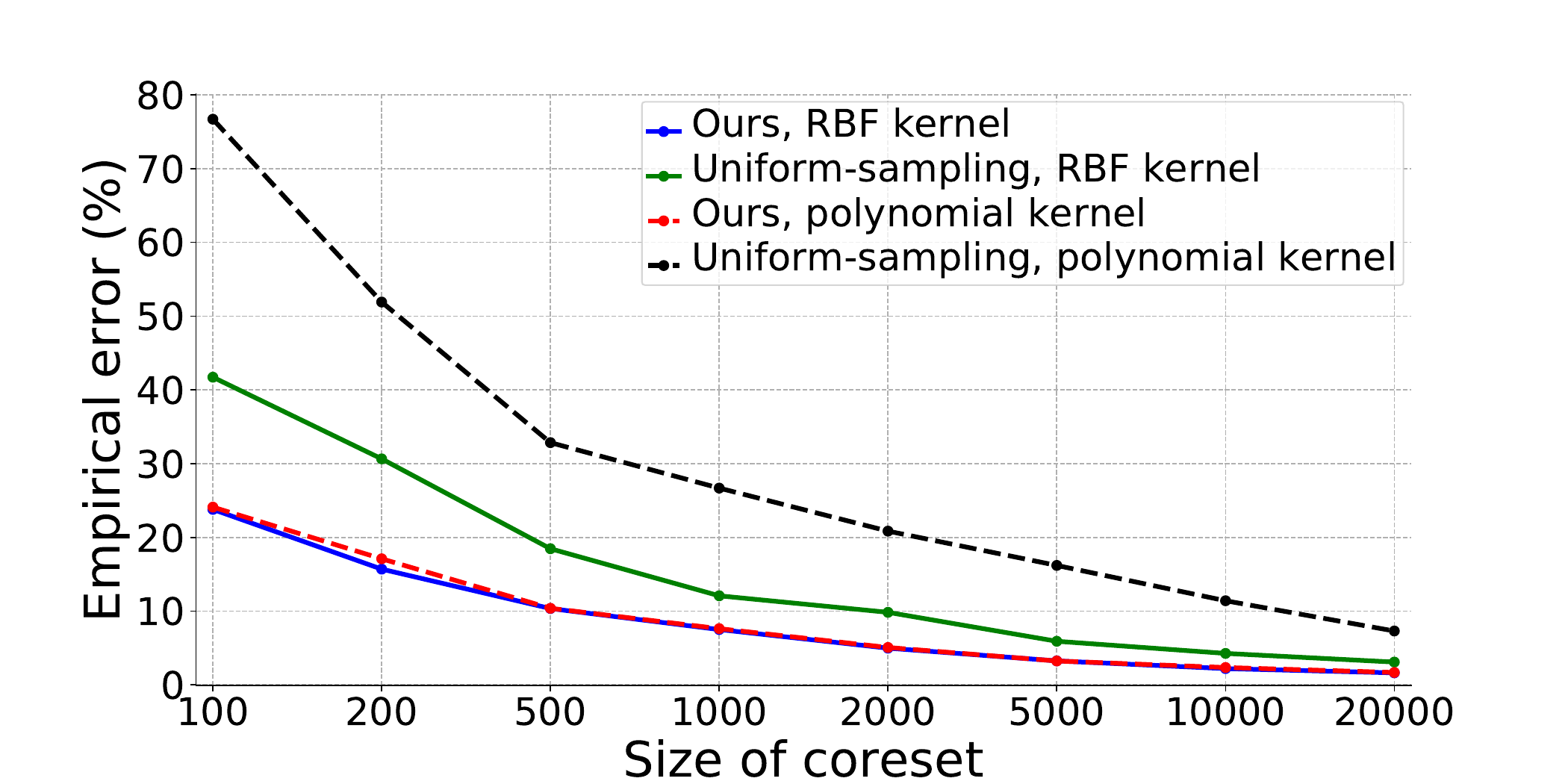}
        \caption*{Bank dataset}
    \end{subfigure}
    \begin{subfigure}[b]{0.39\textwidth}
        \centering
        \includegraphics[width=\textwidth]{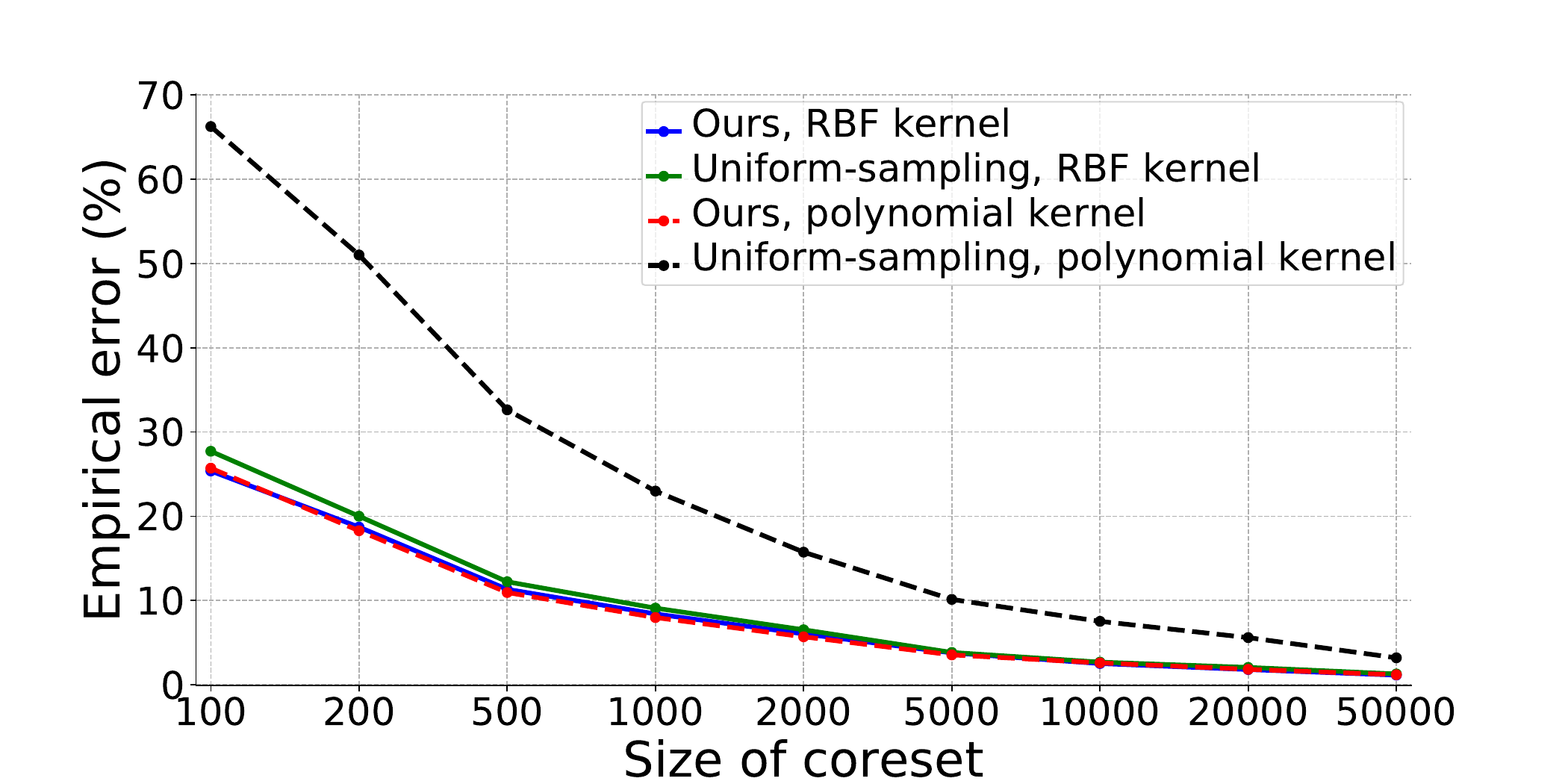}
        \caption*{Twitter dataset}
    \end{subfigure} \qquad
    \begin{subfigure}[b]{0.39\textwidth}
        \centering
        \includegraphics[width=\textwidth]{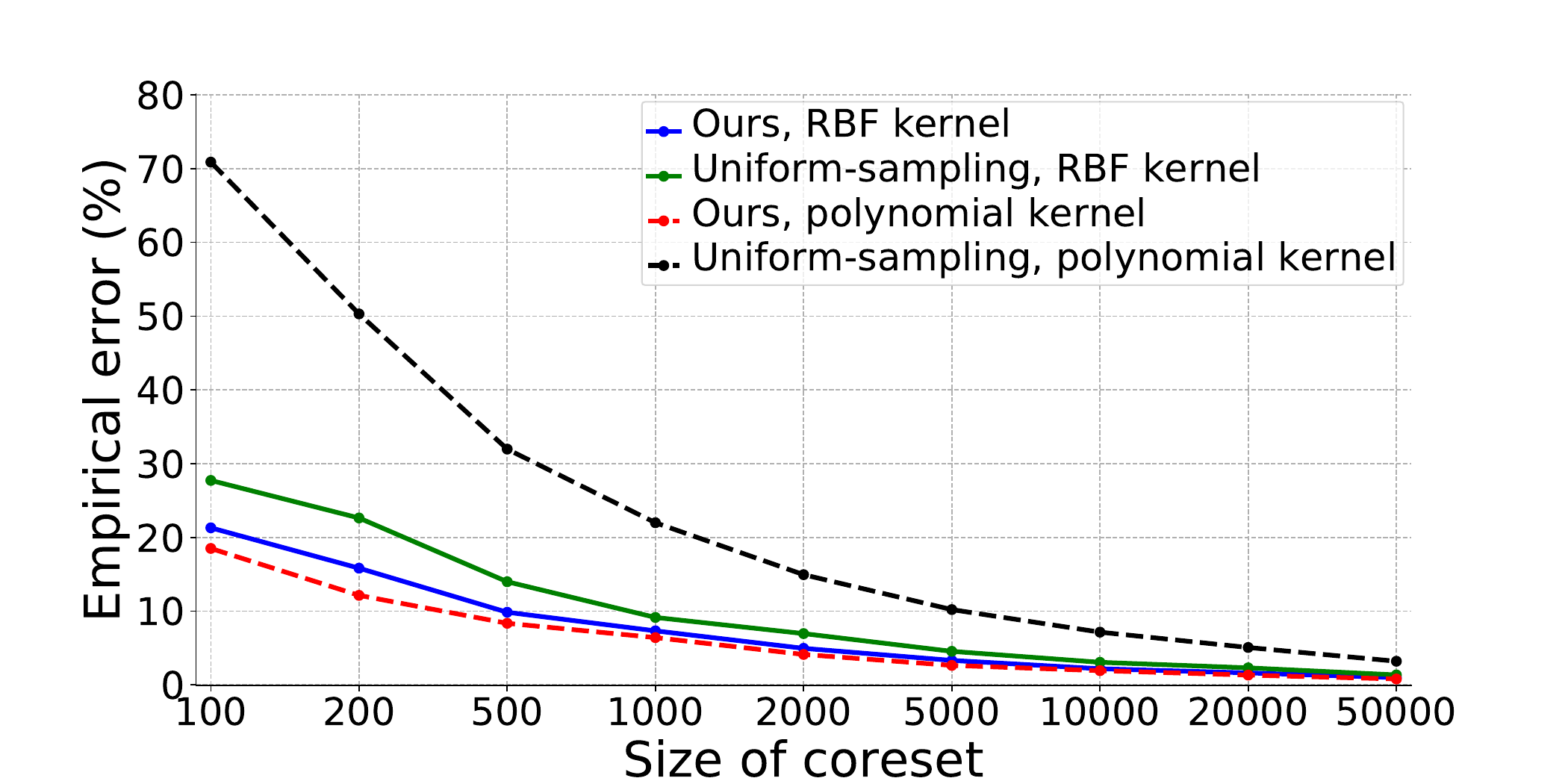}
        \caption*{Census1990 dataset}
    \end{subfigure}
    \caption{Tradeoffs between coreset size and empirical error.}
    \label{fig:size_vs_error}
\end{figure}

\subsection{Speeding Up Kernelized \kMeanspp}
\kMeanspp~\cite{DBLP:conf/soda/ArthurV07} is a widely-used algorithm for \kMeans,
and it could easily be adapted to solve kernel \kMeans.
In particular, \emph{kernelized} \kMeanspp is the \kMeanspp algorithm applied on the feature space $\calH$,
where the distances are the kernel distances $\dist(\varphi(x), \varphi(y))$.
This kernel distance between a pair of points may be evaluated
using the kernel trick efficiently,
however, as mentioned earlier,
implementing \kMeanspp using the kernel trick can take $\Omega(n^2)$ time in total.
We use our coreset to speed up this kernelized \kMeanspp algorithm, 
by first computing the coreset and then running kernelized \kMeanspp on top of it;
this yields an implementation of kernelized \kMeanspp with near-linear 
running time. 

In Figure~\ref{fig:kmeanspp}, we demonstrate the running time and
the error achieved by kernelized \kMeanspp with and without coresets, 
experimented with varying coreset sizes.
We measure the relative error of our coreset-based kernelized \kMeanspp
by comparing the objective value it achieves
with that of vanilla (i.e., without coreset) kernelized \kMeanspp.
These experiments show that the error decreases significantly as the coreset size increases,
and it stabilizes around size $N=100$, achieving merely $<5\%$ error.
Naturally, the running time of our coreset-based approach increases with the coreset size,
but even the slowest one is still several orders of magnitude faster than vanilla kernelized \kMeanspp.

\begin{figure}[t]
    \centering
    \begin{subfigure}[b]{0.39\textwidth}
        \centering
        \includegraphics[width=\textwidth]{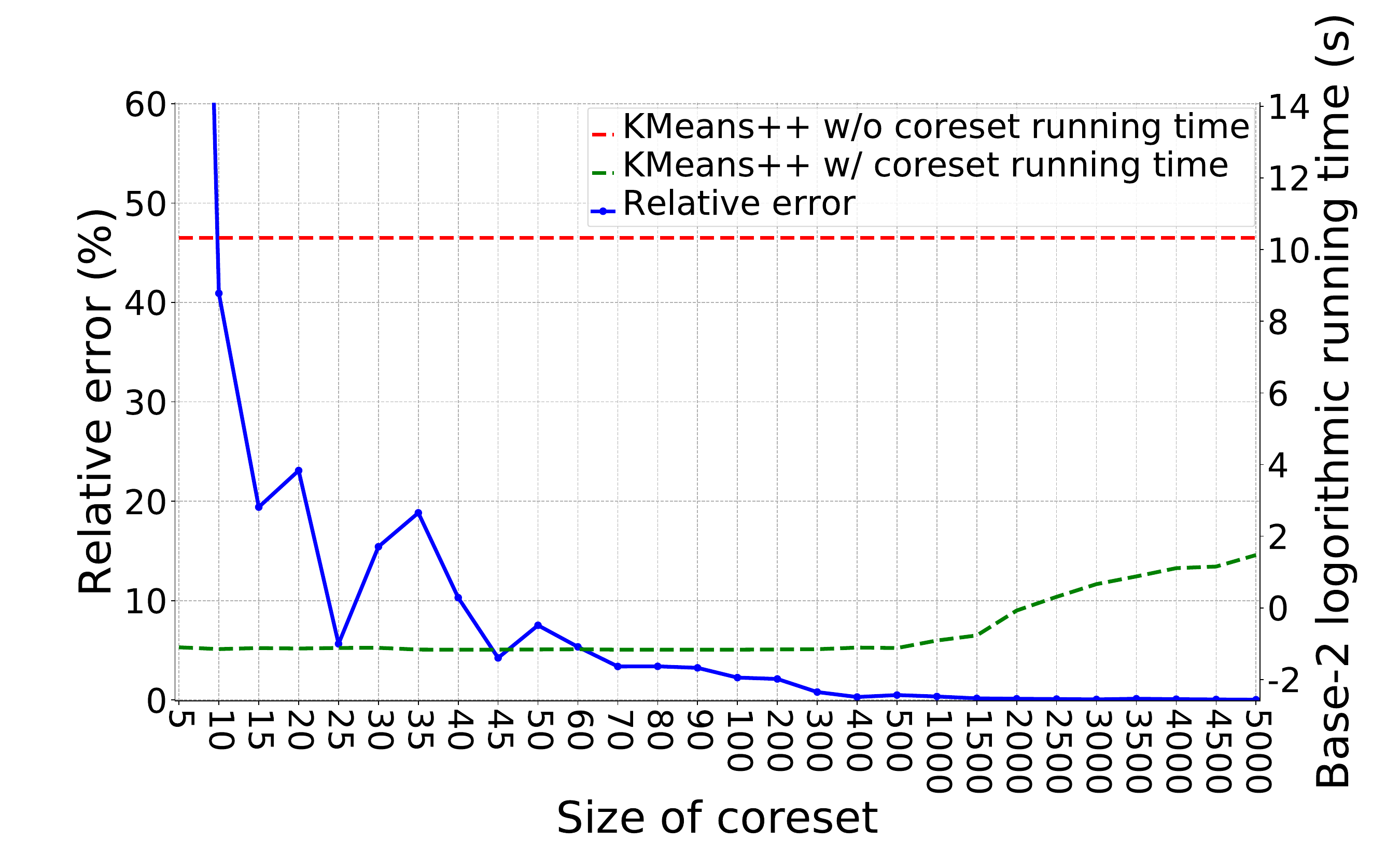}
        \caption*{RBF kernel}
    \end{subfigure} \qquad
    \begin{subfigure}[b]{0.39\textwidth}
        \centering
        \includegraphics[width=\textwidth]{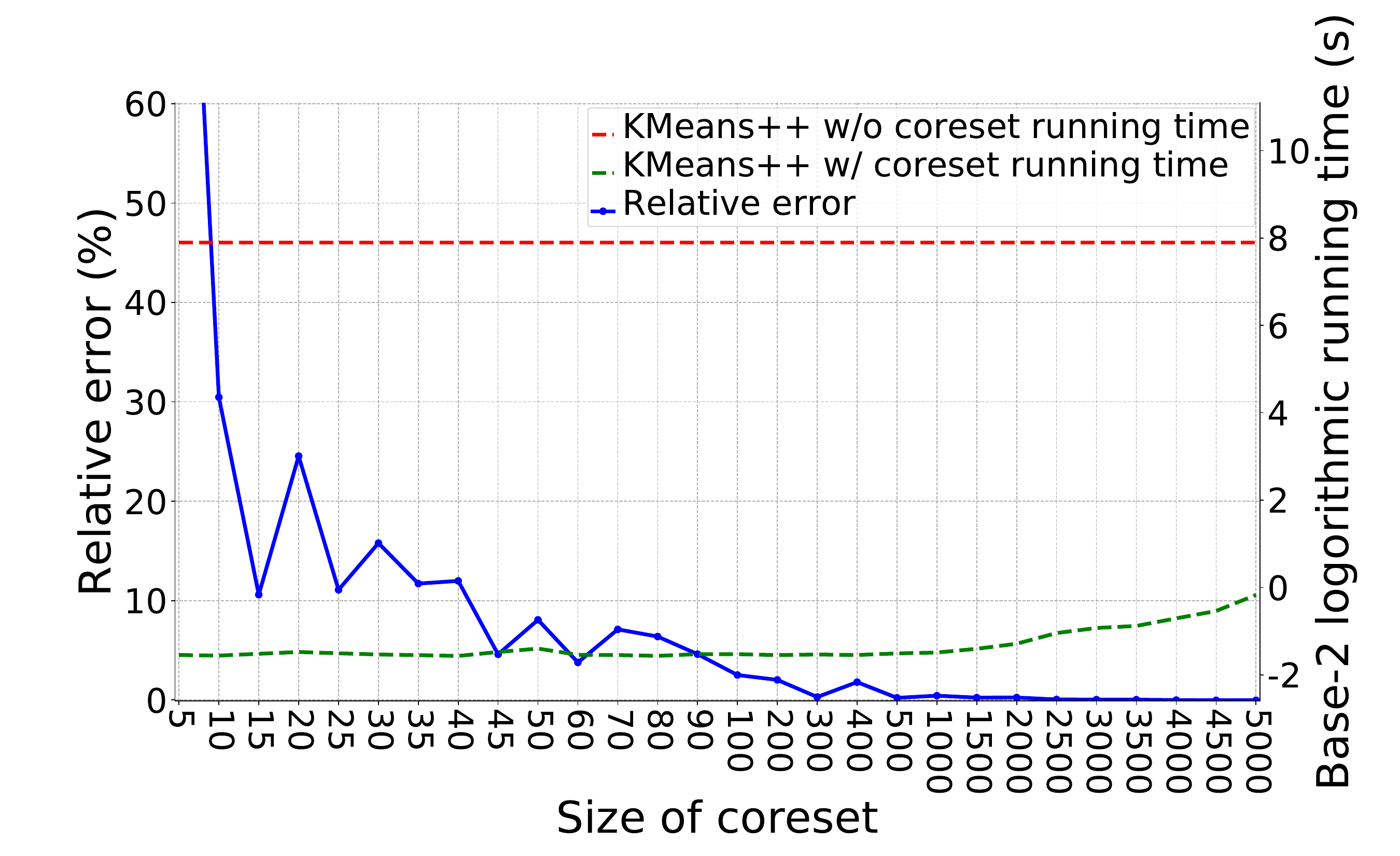}
        \caption*{Polynomial kernel}
    \end{subfigure}
    \caption{
        Speedup of kernelized \kMeanspp using our coreset.
    This experiment is conducted on the Twitter dataset with RBF and polynomial kernels.
We run each algorithm 10 times, and report the average running time and the minimum objective value (in relative-error evaluation). 
    }
    \label{fig:kmeanspp}
\end{figure}

\subsection{Speeding Up Spectral Clustering}
In the spectral clustering problem, the input is a set of $n$ objects $X$ and
an $n \times n$ matrix $A$ that measures the similarity between every pair of elements in $X$,
and the goal is to find a $k$-partition of $X$ such that a certain objective function with respect to $A$ is minimized.
\cite{DBLP:conf/kdd/DhillonGK04} shows a way to write spectral clustering 
as a (weighted) kernel \kMeans problem,
and use the kernel \kMeans solution to produce a spectral clustering.
Specifically, let $D$ be an $n \times n$ diagonal matrix such that
$D_{ii} = \sum_{j \in [n]}A_{ij}$. Then according to~\cite{DBLP:conf/kdd/DhillonGK04},
spectral clustering can be written as a weighted kernel \kMeans problem
with weights $w_i := D_{ii}$ and kernel function $K := D^{-1} A D^{-1}$,
provided that $A$ is positive semidefinite (which could be viewed as a kernel).
We apply this reduction, and use the abovementioned coreset-based kernelized \kMeanspp to solve kernel \kMeans.
We experiment on the subsampled Twitter dataset with varying number of points,
and we use the polynomial and RBF kernels as the similarity matrix $A$.

However, it takes $\Theta(n^2)$ time if we evaluate $D_{ii}$ naively.
To resolve this issue, we draw a uniform sample $S$ from $[n]$,
and use $\hat{D}_{ii} := \frac{n}{\vert S\vert } \sum_{j \in S}{A_{ij}}$ as an estimate for $D_{ii}$.
The accuracy of $\hat{D}$ is justified by previous work on kernel density estimation~\cite{DBLP:conf/compgeom/JoshiKPV11},
and for our application we set $\vert S\vert = 1000$ which achieves good accuracy.

We compare our implementation with the spectral clustering solver
from the well-known \sklearn library~\cite{sklearn}. 
The experimental results, reported in Figure~\ref{fig:spectral},
show that our approach has a much better asymptotic growth of running time than that of \sklearn's,
and achieves more than 1000x of speedup already for moderately large datasets ($n = 40000$). 
Although the difference of absolute running time might be partially caused by efficiency issues of the Python language
used in the \sklearn implementation (recall that our implementation is in C++),
the outperformance in asymptotic growth suggests that our improvement in running time is fundamental, 
not only due to the programming language.
We also observe that our approach yields better objective values than \sklearn.
One possible reason is that \sklearn might be conservative in
using more iterations to gain better accuracy,
due to the expensive computational cost that we do not suffer.
We also observe that our approach yields better objective values than \sklearn.
One possible reason is that \sklearn might be conservative in
using more iterations to gain better accuracy,
due to the expensive computational cost that we do not suffer.

\begin{figure}[t]
    \centering
    \begin{subfigure}[b]{0.39\textwidth}
        \centering
        \includegraphics[width=\textwidth]{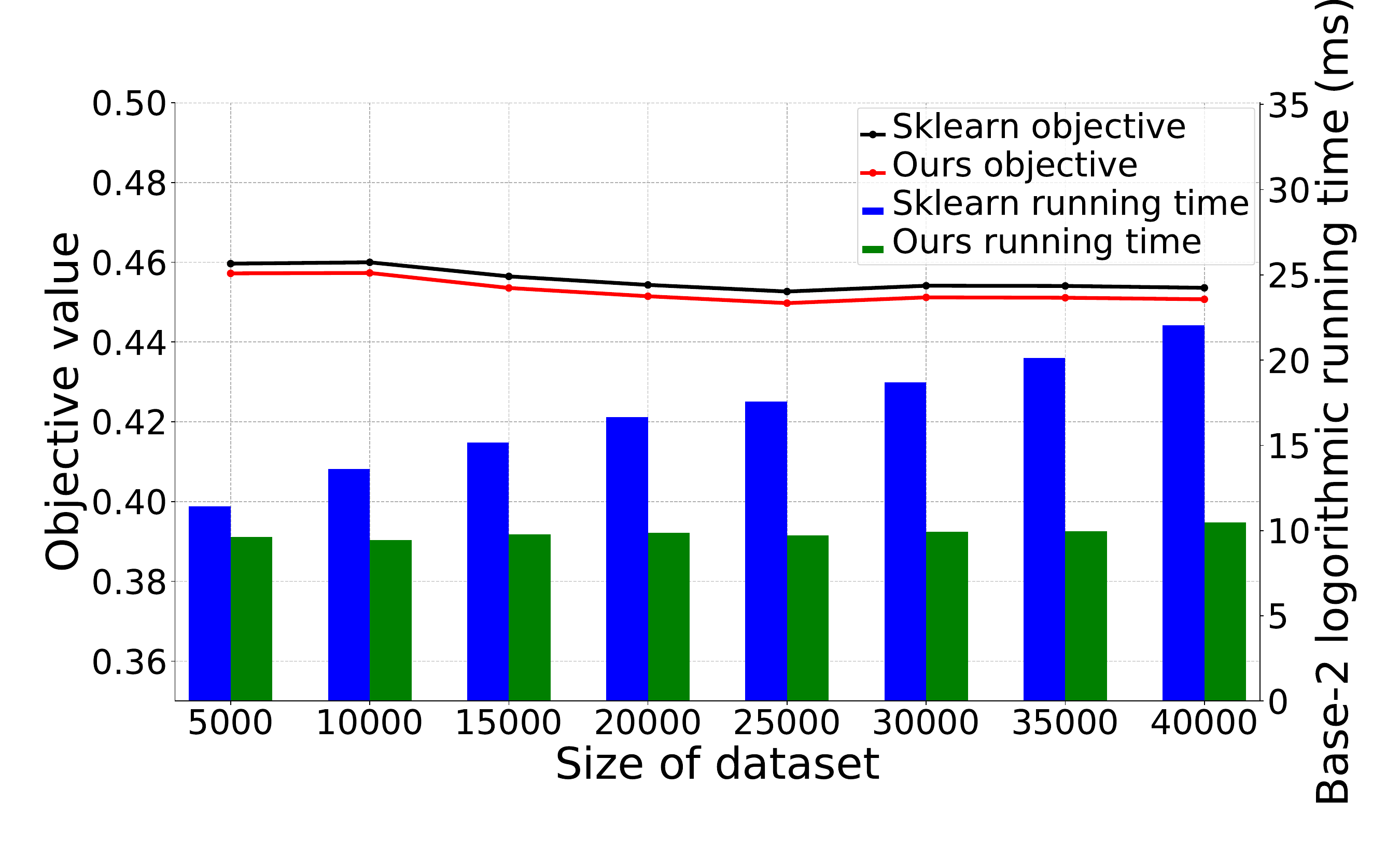}
        \caption*{RBF kernel}
    \end{subfigure} \qquad
    \begin{subfigure}[b]{0.39\textwidth}
        \centering
        \includegraphics[width=\textwidth]{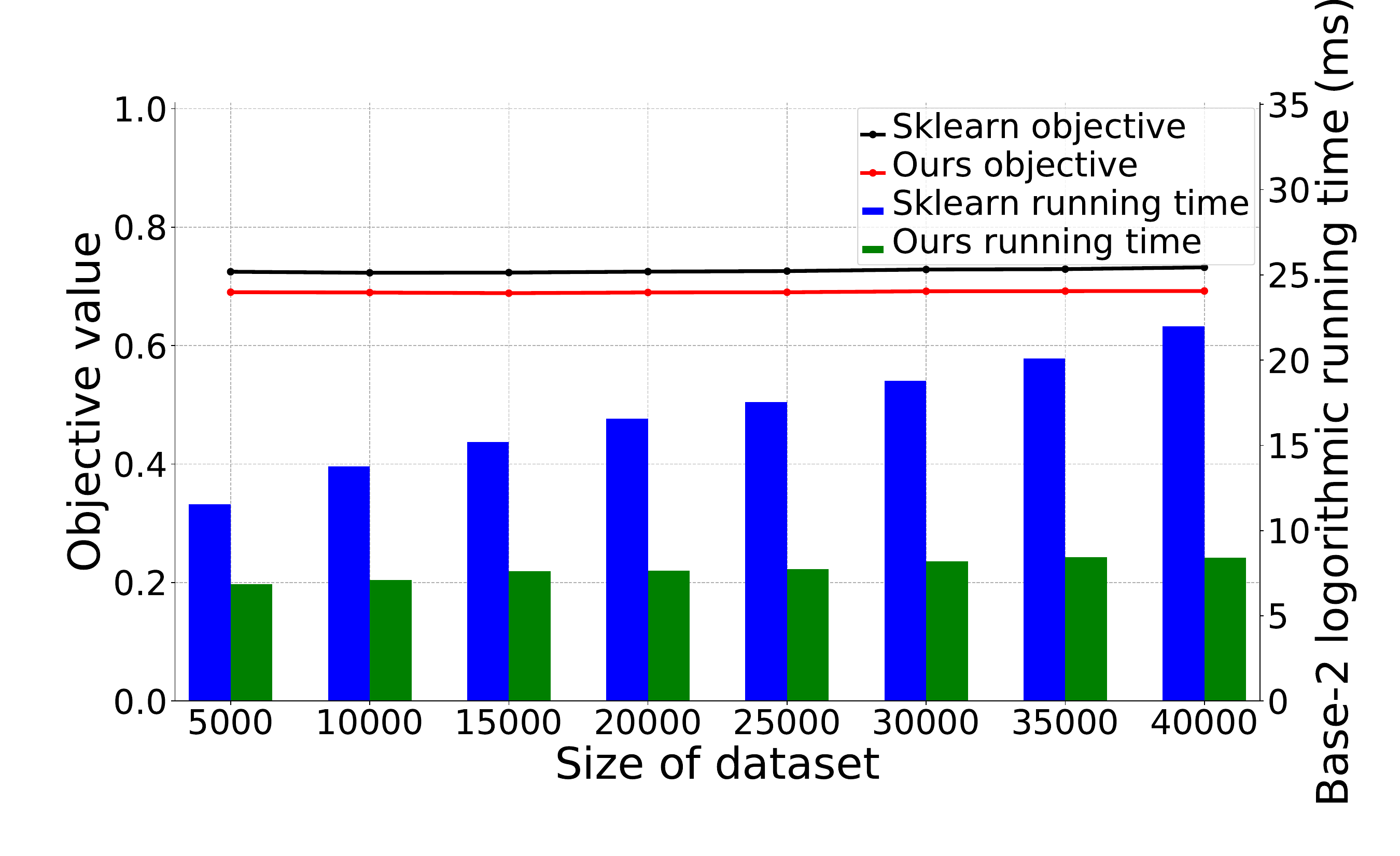}
        \caption*{Polynomial kernel}
    \end{subfigure}
    \caption{Speedup of spectral clustering using coreset-based kernelized \kMeanspp, with coreset size $N=2000$.
      Similar to Figure~\ref{fig:kmeanspp}, we run each algorithm 10 times,
      report the average running time and the minimum objective value.}
    \label{fig:spectral}
\end{figure}

\bibliographystyle{alphaurl}
\bibliography{ref}

\end{document}